\documentclass[prl, twocolumn, hidelinks]{revtex4-1}

\usepackage{graphicx}
\usepackage{epsfig,color}
\usepackage{amsmath,amssymb,eufrak}
\usepackage{amsthm}
\usepackage{enumerate}
\usepackage{subcaption}
\usepackage{color,amsxtra,amsfonts,bm}
\usepackage{hyperref}

\usepackage[shortlabels]{enumitem}
\usepackage{tikz}
\usetikzlibrary{arrows.meta}

\pgfdeclarelayer{nodelayer}
\pgfdeclarelayer{edgelayer}
\pgfsetlayers{edgelayer,nodelayer,main}

\newtheorem{thm}{Theorem}

\newtheorem{corollary}{Corollary}

\def\Id{{\openone}}
\newcommand{\be}{\begin{equation}}
\newcommand{\ee}{\end{equation}}
\newcommand{\bea}{\begin{eqnarray}}
\newcommand{\eea}{\end{eqnarray}}

\newcommand{\bra}[1]{\langle #1|}
\newcommand{\ket}[1]{|#1\rangle}
\newcommand{\one}{\mbox{$1 \hspace{-1.0mm}  {\bf l}$}}

\newcommand{\Nto}{\overset{N}{\rightarrow}}
\newcommand{\Nsim}{\overset{N}{\sim}}

\begin{document}

\title{Matrix Product States: Entanglement, symmetries, and state transformations}

\date{\today}
\author{David Sauerwein}
\email{sauerwein.david@gmail.com}
\affiliation{Max-Planck-Institut f\"{u}r Quantenoptik, Hans-Kopfermann-Str. 1,
85748 Garching, Germany}
\affiliation{Institute for Theoretical Physics, University of Innsbruck, Technikerstr. 21A, 6020 Innsbruck, Austria}

\author{Andras Molnar}
\email{andras.molnar@mpq.mpg.de }
\affiliation{Max-Planck-Institut f\"{u}r Quantenoptik, Hans-Kopfermann-Str. 1,
85748 Garching, Germany}
\affiliation{Munich Center for Quantum Science and Technology (MCQST), Schellingstr. 4, D-80799 M\"{u}nchen, Germany}

\author{J. Ignacio Cirac}
\email{ignacio.cirac@mpq.mpg.de }
\affiliation{Max-Planck-Institut f\"{u}r Quantenoptik, Hans-Kopfermann-Str. 1,
85748 Garching, Germany}
\affiliation{Munich Center for Quantum Science and Technology (MCQST), Schellingstr. 4, D-80799 M\"{u}nchen, Germany}

\author{Barbara Kraus}
\email{barbara.kraus@uibk.ac.at}
\affiliation{Institute for Theoretical Physics, University of Innsbruck, Technikerstr. 21A, 6020 Innsbruck, Austria}

\begin{abstract}
We analyze entanglement in the family of translationally-invariant matrix product states (MPS). We give a criterion to determine when two states can be transformed into each other by SLOCC transformations, a central question in entanglement theory. We use that criterion to determine SLOCC classes, and explicitly carry out this classification for the simplest, non-trivial MPS. We also characterize all symmetries of MPS, both global and local (inhomogeneous). We illustrate our results with examples of states that are relevant in different physical contexts. 
\end{abstract}

\maketitle

{\em 1. Introduction:}
Entanglement is a resource for numerous striking applications of quantum information theory \cite{Nielsen2010,Horodecki2009}. Furthermore, it is key to comprehend many peculiar properties of quantum many-body systems \cite{Amico2008, Eisert2010} and has become increasingly important in areas like quantum field theory or quantum gravity \cite{Nishioka2018, Witten2018}. Despite its relevance, entanglement is far from being fully understood; at least in the multipartite setting. State transformations play a crucial role as they define a partial order in the set of states. For instance, if a state $\Psi$ can be transformed into a state $\Phi$ deterministically by local operations and classical communication (LOCC), then $\Psi$ is at least as entangled as $\Phi$ \cite{Horodecki2009}. If two states cannot even probabilistically be interconverted via local operations, i.e., by so-called stochastic LOCC (SLOCC) transformations, their entanglement is not comparable, as one or the other may be more useful for different informational tasks \cite{Bennett2000}. Thus, the study of state transformations is crucial for the theory of entanglement.

For bipartite systems, state transformations are fully characterized and have led to a very clear picture \cite{Nielsen1999, Vidal1999}, which is widely used in different areas of research. For more parties such a characterization is much more challenging. In general, there are infinitely many classes of states that can be interconverted via SLOCC, and only in few cases they can be characterized, like for symmetric states or for certain tripartite and four-partite states \cite{Bastin2015, Migda2013, Dur2000, Verstraete2002, Briand2004,Chitambar2010}. Moreover, for more than four parties of the same local dimensions almost no state can be transformed into an inequivalent state via deterministic LOCC, and the partial order induced by LOCC becomes trivial \cite{Gour2017,Sauerwein2018}. This shows that generic states are not very interesting from the perspective of local transformations. Additionally, most of the states in the Hilbert space cannot be reached in polynomial time even if constant-size nonlocal gates are allowed \cite{Poulin2011}. Hence, the study of state transformations can be reduced to families of non-generic, but physically relevant states.

In this Letter we present a systematic investigation of state transformations for Matrix Product States (MPS) that describe translationally invariant systems (with periodic boundary conditions) \cite{Fannes1992, Perez-Garcia2007}. Ground states of gapped 1D local Hamiltonians or states generated sequentially by a source can be efficiently approximated by MPS \cite{Hastings2007, Schon2005}. Hence, these states play a very important role in both, quantum information theory and in many-body physics. Despite the fact that they describe a broad variety of phenomena, they have a simple description: tripartite states -- the fiducial states of MPS -- completely characterize the MPS. We give a criterion to estipulate when an SLOCC transformation between two such MPS exists, and further give criteria to determine the SLOCC classes dictated by such a relation. These classes build a finer structure on top of the SLOCC classification of the fiducial states, with the additional structure depending on the system size.

The methods introduced here also allow us to identify all local symmetries of MPS (not only corresponding to unitary representations \cite{Sanz2009, Singh2010})\footnote{By local (global) symmetries we mean (in-)homogeneous symmetries, i.e., that a different (the same) action operates on each physical system.}. This is interesting on its own right in the theory of tensor networks, as it induces a classification of zero temperature phases of matter \cite{Chen2011,Schuch2011, Pollmann2012}. As we show, the problems we address can be mapped to finding out certain cyclic structures of operators acting on tripartite states. Thus, our results allow to answer questions like:
Can an AKLT state be transformed into a cluster state? What are all the symmetries of these states? What are the SLOCC classes of MPS? As we show, the answers to these questions can be strongly size-dependent.

{\em 2. Matrix Product States:} We consider here a chain of $N$ $d$-level systems in a translationally invariant MPS. One such state, $\Psi(A)$, is defined in terms of a tripartite fiducial state,
 \begin{equation}
|A\rangle =  \sum_{j=0}^{d-1} \sum_{\alpha,\beta=0}^{D-1} A^j_{\alpha,\beta} |j,\alpha,\beta\rangle
 \end{equation}
as
 \be
 \label{MPS}
 |\Psi(A)\rangle = \sum_{j_1,\ldots,j_N} {\rm Tr}(A^{j_1}\ldots A^{j_N})|j_1,\ldots,j_N\rangle.
 \ee
Here, $D$ denotes the bond dimension and $A^j$ a matrix with components $A^j_{\alpha,\beta}$. The corresponding tensor is called {\em injective} if those matrices span the set of $D\times D$ matrices. The matrix $\mathcal{A} = \sum_{j,\alpha,\beta} A_{\alpha,\beta}^j \ket{j}\bra{\alpha,\beta}$ then has a left inverse $\mathcal{A}^{-1}$ \cite{Perez-Garcia2007}. This does not occur generically since it requires that $d\ge D^2$. We consider here {\em normal} tensors instead, which are generic and are those that become injective after blocking $L\le 2D^2(6 + \log_2(D))$ sites \cite{Michalek2018}. Furthermore, we consider $N\ge 2L +1$, so that we can apply the fundamental theorem of MPS \cite{Molnar2018}. We call ${\cal N}_{N,D}$ the set of normal, translationally invariant MPS with bond dimension $D$ and $N\ge 2L+1$ sites. Note also that we only consider states with full local ranks as we could otherwise map the problem to smaller local dimensions.

We use several examples of some particularly relevant states of bond dimension 2 to illustrate our results. They are generated by fiducial states $\ket{X_b} = (\Id \otimes b \otimes \Id)\ket{X}$, where $X$ is one of the following states:
    \begin{enumerate}[(i),nosep,align = left, leftmargin =* , widest* = 3]
    	\item the W state $\ket{W} = \ket{100}+\ket{010}+\ket{001}$;
    	\item  the GHZ state $\ket{GHZ} = \ket{000}+\ket{111}$;
    	\item the cluster state $\ket{GHZ_H}$, where $H = \sum_{ij} (-1)^{ij} \ket{i}\bra{j}$;
    	\item  the state $\ket{A_{A}}=\sqrt{2} \ket{010} - \ket{100} + \ket{111} - \sqrt{2} \ket{201}$, which generates the AKLT state;
    	\item  the state $\ket{VB} = \sum_{ij} \ket{k_{ij}ij}$ with $d=4$ and $k_{ij}=2i+j$ generating the valence bond state.
    \end{enumerate}
	The W and GHZ states play a central role in entanglement theory \cite{Dur2000, Greenberger1989}, the cluster state in measurement-based quantum computation \cite{Raussendorf2001}, and the AKLT \cite{Affleck1987} and the valence bound state are paradigmatic examples that appear in condensed matter physics. The latter is, furthermore, injective and the fixed point of a renormalization procedure \cite{Verstraete2005}.\\

{\em 3. Symmetries:} Global symmetries, of the form $u^{\otimes N}$, of MPS were considered in \cite{Sanz2009, Singh2010}, and have led to the classification of phases of MPS in spin chains \cite{Chen2011, Schuch2011, Pollmann2012}. Here we extend those results in two ways by considering: (i) non-unitary symmetries and (ii) local symmetries for which the operators acting on different spins can be different \cite{Note1}. That is, given $\Psi(A)\in{\cal N}_{N,D}$ we look for all operators $g=\bigotimes_{j=1}^N g_j$ such that $|\Psi(A)\rangle=g |\Psi(A)\rangle$.

In order to solve this problem, we define
 \be
 G_A =\{ h=g\otimes x\otimes y^T \ | \ h|A\rangle =|A\rangle \}
 \ee
where $T$ denotes the transponse in the standard basis. We say that $h_1,h_2\in G_A$ with $h_i=g_i\otimes x_i\otimes y_i^T$, can be concatenated and write $h_1\to h_2$ if $y_1 x_2\propto \Id$. For $k \in \mathbb{N}$ we call a sequence $\{h_i\}_{i=1}^k \subseteq G_A$ with
 \be
 \label{concat}
 h_1\to h_2\to \ldots \to h_k\to h_1
 \ee
 a $k-$cyle.
Then we have:

\begin{thm} The local (global) symmetries of $\Psi(A)\in {\cal N}_{N,D}$ are in one-to-one correspondence with the $N$-cycles (1-cycles) in $G_{A}$.
\label{Thm1}
\end{thm}

The symmetry of the state corresponding to the cycle $h_1\to h_2\to \ldots \to h_N\to h_1$ is $g_1\otimes \dots \otimes g_N$. The trivial symmetry with $g=\Id$ always exists.
The proof is based on the fundamental theorem of MPS \cite{Molnar2018} and is given in the Supplemental Material (SM) \cite{supp}. Hence, one simply has to determine $G_A$ and find all of its $N$-cycles to characterize the local symmetries of $\Psi(A)$. It suffices to find all minimal cycles of $G_{A}$ from which all others can be obtained by concatenation. For example, a 2-cycle can always be concatenated with itself to an $N$-cycle if $N$ is even. The global symmetries are defined in terms of 1-cycles, and thus require $g\otimes x^{-1} \otimes x^T |A\rangle=|A\rangle$.
For $g$ unitary we therefore recover previous results \cite{Sanz2009, Singh2010}. The novelty relies on the fact that one may also have local symmetries, with different $g_j$. In the following we illustrate this fact and the dependence of the symmetries on the system size.

For injective MPS with $D = d^2$ it is straightforward to show that \cite{supp}
\be
G_A = \{s_{x,y} \otimes x \otimes y^T \ | \ x,y \in GL(D, \mathbb{C})\}, \label{eq:injsym}
\ee
where $s_{x,y} = \mathcal{A}({x^T}^{-1} \otimes y^{-1})\mathcal{A}^{-1}$. These operators can be concatenated to infinitely many cycles of arbitrary length. The corresponding symmetries are parametrized via regular matrices $x_1,\ldots,x_{N}$ as
\be
S(x_1,\ldots,x_{N}) = s_{x_N^{-1},x_1} \otimes \ldots \otimes s_{x_{N-1}^{-1},x_N}.
\ee
For $\mathcal{A} = \one$ we obtain the large local symmetry group of the injective valence bond state.
Normal (but not injective) states have a much smaller set of symmetries. For the AKLT state 
\be
G_{A_{A}} = \{s_x \otimes x^{-1} \otimes x^T \ | \ x \in GL(2, \mathbb{C})\}, \label{eq:AKLTsym}
\ee
where $s_x$ is a function given in the SM \cite{supp}. Clearly, elements of $G_{A_A}$ can only be concatenated with themselves. Consequently, the local symmetry group of the AKLT state possesses only global symmetries of the form $s_x^{\otimes N}$. Moreover, this group is isomorphic to the projective linear group $PGL(2, \mathbb{C})$ and includes the well-known symmetries with $s_x \in SO(3)$ \cite{Affleck1987}. For the AKLT-type states we have $G_{A_{A,g}} = (\Id \otimes g \otimes \Id)G_A(\Id \otimes g^{-1} \otimes \Id)$ and the local symmetries of $\Psi(A_{A,g})$ read
\be
S(x) = s_x \otimes s_{g^{-1}xg} \otimes \ldots \otimes s_{g^{-(N-1)}xg^{N-1}},
\ee
where $x$ is such that $g^{-N}xg^{N} = x$. Hence, $G_{A_{A,g}}$ is generically smaller than $G_{A_{A}}$. Moreover, it consists of non-global symmetries which are $N$-dependent. The symmetries of the cluster state are also non-global and coincide with the $2^N$ so-called stabilizer symmetries \cite{Hein2006, supp}. For W-generated states the set $G_{A_{W,g}}$ contains infinitely many elements. However, the only nontrivial minimal cycles are 2-cycles. Hence, the corresponding MPS has only the trivial symmetry for odd $N$ and infinitely many non-translationally invariant symmetries for even $N$.

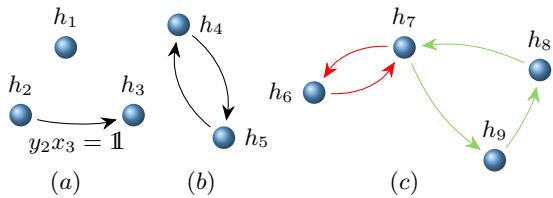
\begin{figure}
\begin{tikzpicture}[scale=0.6]
    \definecolor{barbarablue}{HTML}{60A2DEFF}
    \definecolor{barbaragreen}{HTML}{8DD35FFF}
    % styles in the figure
	\tikzset{
	    >={Stealth[width=1.7mm,length=2.1mm]},
		dot/.style={
			inner sep = 0pt,
			minimum size = 0.3cm,
			shape = circle,
			shade,
			ball color=barbarablue,
		},
		none/.style={},
		edge/.style={
			->,
			shorten <=2pt,
			shorten >=2pt,
		},
	}
	% a single symmetry 
	\node [style=dot, label={above:$h_1$}] (alone) at (-3.5, 1) {};
	% a connection
	\node [style=dot, label={above:$h_2$}] (connect-1) at (-4.5, -0.5) {};
	\node [style=dot, label={above:$h_3$}] (connect-2) at (-2, -0.5) {};
    \draw [style=edge, bend right = 15, looseness=0.75] (connect-1) to node [midway,below] {$y_2 x_3=\Id$} (connect-2) ;
	\node [style=none] (14) at (-3.5, -2) {$(a)$};
    % a two-cycle
	\node [style=dot, label={right:$h_4$}] (cycle-1) at (-1, 1.5) {};
	\node [style=dot, label={right:$h_5$}] (cycle-2) at (0, -1) {};    
	\draw [style=edge, bend left] (cycle-1) to (cycle-2);
	\draw [style=edge, bend left] (cycle-2) to (cycle-1);
	\node [style=none] (15) at (-0.5, -2) {$(b)$};
    % two-and three cycle together
	\node [style=dot, label={left:$h_6$}] (complex-1) at (2, 0) {};
	\node [style=dot, label={above:$h_7$}] (complex-mid) at (4, 1) {};
    
	\node [style=dot, label={above:$h_8$}] (1) at (7, 0.5) {};
	\node [style=dot, label={above:$h_9$}] (2) at (6, -1.5) {};

	\draw [style=edge, red,bend right] (complex-mid) to (complex-1);
	\draw [style=edge, red,bend right] (complex-1) to (complex-mid);
	\draw [style=edge, barbaragreen,bend right=15] (complex-mid) to (2);
	\draw [style=edge, barbaragreen,bend right=15] (1) to (complex-mid);
	\draw [style=edge, barbaragreen,bend right=15, looseness=0.75] (2) to (1);

	\node [style=none] (16) at (4, -2) {$(c)$};
\end{tikzpicture}

\caption{ Graphical representation of how operators in
$G_A$ (balls) can be concatenated (edges). (a) $h_1$ cannot be concatenated, $h_2$ only with $h_3$ meaning $y_2x_3\propto \Id$; (b) $h_4$ and $h_5$ form a $2$-cycle; (c) two minimal cycles sharing an operator.
}\label{fig1}
\end{figure}

{\em 4. State transformations:} Here, we answer the question of when
a state, $\Psi(A)$, can be converted into another one, $\Psi(B)$, by SLOCC. As both states correspond to normal tensors, they both belong to ${\cal N}_{N,D}$. We write $A \Nto B$ if the transformation is possible. Note that
  \begin{equation}
    A \Nto B \quad  \text{iff} \quad  \bigotimes_{j=1}^N g_j \ket{\Psi_N(A)} = \ket{\Psi_N(B)}
  \end{equation}
for some $g_j$. We distinguish also here between global (where all $g_j$ are equal) and local transformations. As for symmetries, we define the set
 \begin{equation}
    G_{A,B} = \left\{ h= g\otimes x \otimes y^T \ | \ h \ket{A} = \ket{B} \right\},
 \end{equation}
where $x,y$ are regular, but not necessarily $g$. That is, we also consider the case where the physical dimensions, $d_A$ and $d_B$ do not coincide. 
It is straightforward to show that if $|B\rangle = h_0 |A\rangle$ with $h_0=g_0\otimes x_0\otimes y_0^T$ then $G_{A,B} \supseteq h_0 G_A$. For $d_A = d_B$, which is the case iff $g_0$ is regular, we have $G_{A,B} = h_0 G_A$. Defining concatenations of elements in $G_{A,B}$  as well as $k$-cycles as before, we have:

\begin{thm} $A \Nto B$ with local (global) transformations iff there exists
	 an $N$-cycle
	(1-cycle) in $G_{A,B}$.
\label{Thm2}
\end{thm}

The proof is given in the SM \cite{supp}. Theorem \ref{Thm2} solves the state transformation problem. We can immediately make some simple statements about different possibilities that may occur. For instance, if $G_{A,B}$ only contains a 1-cycle, then $A \Nto B$ for all $N$ with just global transformations. However, if the only minimal cycle is a 2-cycle, then the transformation can only happen for even $N$. As we illustrate in the following, some transformations might require more sophisticated operations and one obtains a rich variety of situations.

An injective $\Psi(A) \in {\cal N}_{N,D}$ can be transformed to any $\Psi(B) \in {\cal N}_{N,D}$ via the global operation
\be
(\mathcal{B}\mathcal{A}^{-1})^{\otimes N} |\Psi(A)\rangle = |\Psi(B)\rangle.
\ee
However, using that $G_{A,B} \supseteq (\mathcal{B}\mathcal{A}^{-1} \otimes \Id \otimes \Id) G_A$, with $G_A$  given in (\ref{eq:injsym}), we find that $G_{A,B}$ also contains infinitely many $N$-cycles that lead to non translationally invariant operators that transform $|\Psi(A)\rangle$ into $|\Psi(B)\rangle$. As a special case, we obtain the well-known result that the injective valence bond state can be transformed to any MPS in ${\cal N}_{N,2}$. Since injective MPS are generic in $\mathcal{N}_{N,D}$ for $d = D^2$, a randomly selected MPS of these dimensions can be transformed into any other state of the same dimensions. In contrast to that, transformations from normal (but not injective) states are much more restricted. For example, we show below that for $d = D = 2$ any two randomly selected states $\Psi(A), \Psi(B) \in \mathcal{N}_{N,2}$ cannot be transformed into each other for any $N$. For the AKLT state and the cluster state, $G_{A_A,A_{Cl}}$ contains only 2-cycles. Hence,
the AKLT state can be transformed into the cluster state iff $N$ is even (see SM \cite{supp}). The reverse transformation is impossible since the physical dimension cannot be increased. Particularly
sophisticated transformations are necessary to transform the AKLT state into certain AKLT-type states, $\Psi(A_{A,g})$, for which
\be
G_{A_A,A_{A,g}} = \{s_x \otimes gx^{-1} \otimes x^{T} \ | \ x \in GL(\mathbb{C}, 2)\}. \nonumber
\ee
Using Theorem \ref{Thm2} it is easy to show that $A_A \Nto A_{A,g}$ iff $g^N \propto \Id$. Hence, the feasibility of this transformation is highly size dependent. Moreover, for any $M \in \mathbb{N}$ there exists a regular $g$ such that $A_A {\overset{M}{\rightarrow}} A_{A,g}$ is not possible for any $M < N$, but it is for $M = N$.

{\em 5. Equivalence classes under SLOCC transformations:}
SLOCC classes give a coarse but very useful classification of entanglement in many-body systems. We show now how these classes can be obtained for ${\cal N}_{N,D}$.

We write $A\Nsim B$ if $\Psi(A)$ is SLOCC equivalent to $\Psi(B)$, i.e., 
  \begin{equation}
    A \Nsim B \quad \text{iff} \quad  \bigotimes_{j=1}^N g_j \ket{\Psi(A)} = \ket{\Psi(B)}
  \end{equation}
for some regular $g_j$. Thus, we can reduce the study of SLOCC classes to that of the tripartite fiducial states.

In order to simplify this task, let us make some observations. First, $A \Nsim B$ iff $A \Nto B$ and $B \Nto A$. Because of Theorem \ref{Thm2} the equivalence $A\Nsim B$ thus requires $G_{A,B}\neq \emptyset$ and therefore that $|A\rangle$ and $|B\rangle$ themselves belong to the same tripartite SLOCC class. Hence, the equivalence relation $\Nsim$ induces a classification that is finer than the SLOCC classification of tripartite states. Second, for any regular $g$, $g^{\otimes N}|\Psi(A)\rangle$ is obviously in the same class as $|\Psi(A)\rangle$ and for any regular $x$ we trivially have $\ket{A} \Nsim (\Id\otimes x^{-1}\otimes x^T)\ket{A}$, since both states correspond to the same MPS. That is, $A \Nsim B$ trivially holds if the relation $|B\rangle=(g\otimes x^{-1}\otimes x^T)|A\rangle$ holds (i.e., there exists a $1$-cycle). We get rid of this trivial case by restricting the SLOCC classes to the quotient set induced by that relation. Hence, it only remains to consider states of the form $|A_b\rangle=(\Id\otimes b\otimes \Id)|A\rangle$.
This  observation leads to the following procedure to characterize SLOCC classes of normal MPS (see Fig. \ref{fig:2}): (i) for each tripartite SLOCC class, choose a representative, $A$; (ii) consider all states $|A_b\rangle=(\Id\otimes b\otimes \Id)|A\rangle$ corresponding to a normal tensor; (iii) determine the classes among those states according to the relation $\Nsim$. We now show how this procedure can be carried out.

 According to Theorem \ref{Thm2}, $A_b\Nsim A_c$ requires the existence of an $N$-cycle in $G_{A_b,A_c}$ (or, equivalently, in $G_{A_c,A_b}$). The fact that $G_{A_b,A_c}= (\Id\otimes c\otimes \Id)G_A (\Id\otimes b^{-1}\otimes \Id)$ motivates the following definition (analogous to (\ref{concat})). We say that $h_1,h_2\in G_{A}$, with $h_i = g_i \otimes x_i \otimes y_i^T$, can be $(b \to c)$-concatenated, if $y_1 b x_2 \propto c$ and then write $h_1 \xrightarrow{b\to c} h_2$. A sequence $\{h_i\}_{i=1}^k \subseteq G_{A}$ is called a $(b\to c)$-$k$-cycle if
  \begin{equation}\label{eq:bc_cycle}
    h_1 \xrightarrow{b\to c} h_2 \xrightarrow{b\to c} \dots \xrightarrow{b\to c} h_k \xrightarrow{b\to c} h_1.
  \end{equation}
We obtain the following corollary to Theorem \ref{Thm2}.
\begin{corollary}\label{cor:cycle_spec}
  $A_b\Nsim A_c$ holds nontrivially iff there exists a $(b\to c)$-$N$-cycle in $G_A$ with $N>1$, but no $(b\to c)$-$1$-cycle.
\end{corollary}
Note that this corollary requires that the $N$-cycle contains at least two different elements of $G_A$. This fact can be used to simplify the procedure. For instance, if one wants to determine the $A_c$ connected by 2-cycles, one can take arbitrary $\alpha,\beta$ and impose $y_\alpha c x_\beta = b \propto y_\beta c x_\alpha$, from which one can eliminate $b$. Then, the condition can be mapped into the eigenvalue equation $M\vec{c} = \lambda \vec{c}$, where $M= y_{\alpha}^{-1} y_\beta \otimes (x_{\alpha} x_\beta^{-1})^T$. Thus, by choosing all possible pairs in $G_A$ one can identify all classes corresponding to 2-cycles. Corollary \ref{cor:cycle_spec} solves also straightforwardly the equivalence problem of MPS under local unitary operations (see \cite{Verstraete2005} for global unitary operations).

\begin{figure}
\includegraphics[width=0.3\textwidth]{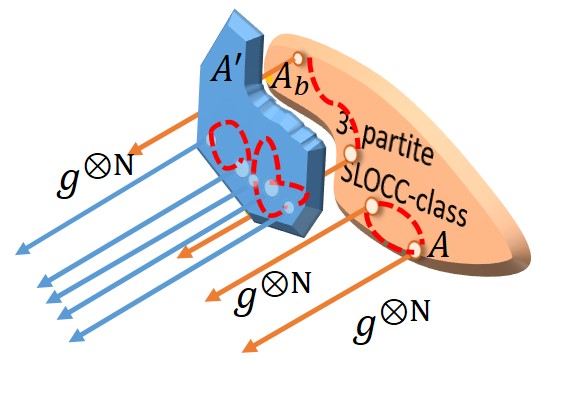}
 \caption{Illustration of SLOCC classes of MPS in accordance with the procedure given in the main text.}
  \label{fig:2}
\end{figure}

The procedures above can be carried out whenever the tripartite SLOCC classes are known, as is the case for $d = 2, D \geq 2$ \cite{Chitambar2010}. Here, we determine the classes for the simplest non-trivial MPS, i.e., those with $d = D=2$ (see Table \ref{tab:one}). The fiducial states are either SLOCC equivalent to the GHZ or the W state \cite{Dur2000}. Hence, the corresponding SLOCC classes separate into GHZ- and W-generated ones. All W-generated MPS are SLOCC equivalent. As explained before, it, hence, remains to consider states of the form $|GHZ_b\rangle = \Id \otimes b \otimes \Id|GHZ\rangle$. The resulting classes can be coarse grained into three sets according to the value of $\chi(b) \equiv \frac{b_{00} b_{11}}{b_{01}b_{10}}$, where $b_{ij}$ denote the entries of $b$: (i) the generic set ($\chi \neq -1,0$); (ii) the cluster set ($\chi = -1$); (iii) the symmetryless set ($\chi = 0$). The generic set is of full measure in the set of all MPS with $d = D = 2$ and is comprised of states whose local symmetries are of the form $\{\Id,s^{\otimes n}\}$. For two such states we have
\begin{align*}
GHZ_b \Nsim GHZ_c \ \Leftrightarrow \chi(b) = \begin{cases} \chi(c) \ \text{or} \ \chi(c)^{-1} \ \text{ and $N$ even}\\
\chi(c) \ \text{and $N$ odd}.\end{cases}
\end{align*}
Thus, there are infinitely many, $N$-dependent classes. The cluster set coincides with the set of states which are SLOCC equivalent to the cluster state. They possess $2^N$ local symmetries. The states in the symmetryless set are also all SLOCC equivalent and have only the trivial local symmetry. Combined with the class of the non-normal MPS generated by the GHZ state these classes constitute the SLOCC classification of entangled MPS with $d = D = 2$.

\begin{table}
 \begin{tabular}{c|c|c|c}
 type & $\chi$ & \# symm. & \# SLOCC classes \\
 \hline
  GHZ & $\neq -1,0$ &$2$ & $\infty$ (see main text) \\ \cline{2-4}
  & $-1$ &$2^N$  & $1$ (cluster set)\\ \cline{2-4}
  & 0 &$1$ & $1$ (symmetryless set)\\\cline{2-4}
 \hline
 W & n.a. &$N$-dependent & $1$\\
 \end{tabular}
 \caption{The SLOCC classification of normal $d = D = 2$ MPS. See the SM \cite{supp} for more details.}
  \label{tab:one}
\end{table}

{\em 6. Summary and outlook:} We solved the problem of when an MPS generated by  a normal tensor can be transformed into another via SLOCC and showed how local symmetries of normal MPS can be characterized. In contrast to other results we considered all, in particular non translationally invariant and non unitary, operations. This revealed interesting features of many, particularly relevant states. Furthermore, we provided a procedure to characterize SLOCC classes of normal MPS and explicitly determined them for $d = D = 2$.
We believe that these results can be extended to non-normal MPS and (certain) Projected Entangled Pair States. Furthermore, one can also determine the SLOCC classes of MPS with higher dimensions and their corresponding symmetries. The theory presented here also serves as a basis to study deterministic LOCC transformations. Finally, our characterization of all local symmetries may be relevant in the study of phases of matter for 1D systems.

{\em Acknowledgments: --- } This  research  was  supported  by the  Austrian  Science  Fund  (FWF)  through  Grants  No. Y535-N16 and No.  DK-ALM:W1259-N27, the Austrian Academy of Sciences' Innovation Fund ``Research, Science and Society'', by the European Research Council through the ERC Advanced Grant QENOCOBA under the EU Horizon 2020 program (grant agreement 742102) and by the Deutsche Forschungsgemeinschaft (DFG, German Research Foundation) under Germany's Excellence Strategy - EXC-2111 - 390814868.

\bibliography{BibLetter}

\end{document}

\title{Supplemental Material}

%\author{David Sauerwein}
%\email{david.sauerwein@uibk.ac.at}
%\affiliation{Institute for Theoretical Physics, University of Innsbruck, 6020 Innsbruck, Austria}

\maketitle

In Section \ref{sec:properties} we first review some properties of MPS that are useful for the remainder of this Supplemental Material (SM). In Section \ref{sec:proofs} we proof Theorem 1 and Theorem 2 of the main text. In Section \ref{sec:syms} we determine the local symmetry groups and in Section \ref{sec:trafos} the transformations of the example states mentioned in the main text. In Section \ref{sec:SLOCC} we derive the SLOCC classification of $d = D = 2$ MPS.
We use the same definitions and notations as in the main text. Moreover, we denote by $\sigma_1,\sigma_2,\sigma_3$ the Pauli matrices and use the notation $\C^\times = \C \setminus \{0\}$.

\section{Injective and normal Matrix Product States}\label{sec:properties}

MPS are defined in terms of rank three tensors. We use the following notation throughout the SM. Let us consider a rank-three tensor $A\in \mathbb{C}^d \otimes \mathbb{C}^D \otimes \mathbb{C}^D$ with
\begin{equation} \label{Eq:tensor}
  A = \sum_{i=0}^{d-1} \sum_{\alpha,\beta = 0}^{D-1} A_{\alpha\beta}^i \ket{i} \ket{\alpha}\bra{\beta}.
\end{equation}
Given the tensor $A$, we write:
\begin{align}
  A^i & = \sum_{\alpha\beta} A_{\alpha\beta}^i \ket{\alpha}\bra{\beta}, \\
  \ket{A} &= \sum_{i\alpha\beta} A_{\alpha\beta}^i \ket{i} \ket{\alpha}\ket{\beta}, \\
  \mathcal{A} &= \sum_{i\alpha\beta} A_{\alpha\beta}^i \ket{i} \bra{\alpha\beta}.
\end{align}
Clearly, the last two tensors and the set of matrices $\{A^i\}$ are equivalent representations of $A$. The state $\ket{A}$ is often referred to as the fiducial state of the tensor. It can also be expressed as
\begin{equation}
 \ket{A} = (\one \otimes \mathcal{A} \otimes \one)\left(\ket{\Phi^+} \otimes \ket{\Phi^+}\right) \equiv \mathcal{A}^{(23)} \left(\ket{\Phi^+} \otimes \ket{\Phi^+}\right) , \label{eq:dec1}
\end{equation}
where $\ket{\Phi^+} = \sum_{\alpha= 0}^{D-1} \ket{\alpha,\alpha}$ is the maximally entangled state.

In this SM, we consider non-translationally invariant (non-TI) MPS on $N$ subsystems that are defined with the help of  $N$ different tensors $A_1, \dots A_N \in \mathbb{C}^d \otimes \mathbb{C}^D \otimes \mathbb{C}^D$ as
\begin{equation}
\ket{\Psi} = \sum_{i_1,\ldots,i_N = 0}^{d-1} \tr\left(A_1^{i_1}\ldots A_N^{i_N}\right) \ket{i_1,\ldots,i_{N}}. \label{eq:MPS}
\end{equation}
An MPS that is generated by a single tensor $A = A_1 = \ldots = A_N$ is TI and denoted by $\Psi(A)$.

A particularly important class of MPS is the one which corresponds to \emph{normal} tensors. A set of tensors $A_1, \ldots, A_N$ as defined in Eq. (\ref{Eq:tensor}) is called  normal if there is an $L$ such that any $L$ consecutive tensors satisfy that the map
\begin{equation}
X \mapsto \sum_{i_1,\ldots,i_{L} = 0}^{d-1} \tr\left(A_k^{i_1}\ldots A_{k+L-1}^{i_{L}} \cdot X\right) \ket{i_1,\ldots ,i_{L}}
\end{equation}
is injective. Here and in the following, all indices are periodical, i.e., $i+N \equiv i$. $L$ is referred to as the injectivity length of the MPS. The normality of a tensor can equivalently be characterized as the property that any $L$ consecutive tensors satisfy
\begin{equation}
\underset{i_1,\dots,i_L}{\mathrm{span}} \left\{ A_k^{i_1}\ldots A_{k+L-1}^{i_{L}} \right\} = \mathbb{C}^{D \times D}.
\end{equation}
The set of normal MPS on $N$ subsystems with bond dimension $D$ is denoted by $\mathcal{N}_{N,D}$. A tensor is called \emph{injective} if it is normal with $L=1$. In a slight abuse of standard notation, we call an MPS normal (injective) if the corresponding tensor is normal (injective) respectively. An other equivalent condition for being injective is that the map $\mathcal{A}$ corresponding to the defining tensor $A$ has a left inverse $\mathcal{A}^{-1}$ such that
\begin{equation}
  \mathcal{A}^{-1} \mathcal{A} = \sum_{\alpha \beta} \ket{\alpha,\beta}\bra{\alpha,\beta} .
\end{equation}
Note that injectivity requires $d \geq D^2$. Since we are only interested in MPS whose single-subsystem reduced states have full rank the only injective MPS we consider satisfy $d = D^2$.

A fundamental property of MPS is that two different sets of tensors can generate the same state. For instance, if the tensors $B_1, \ldots, B_N$ are related to the tensors $A_1, \ldots, A_N$ as $A_k^j = x_k^{-1}B_k^jx_{k+1}$ for all $k,j$, with $x_{N+1} \equiv x_1$, then
\begin{equation}
\sum_{i_1,\ldots,i_N = 0}^{d-1} \tr\left(A_1^{i_1}\ldots A_N^{i_N}\right) \ket{i_1,\ldots,i_{N}} =
\sum_{i_1,\ldots,i_N = 0}^{d-1} \tr\left(B_1^{i_1}\ldots B_N^{i_N}\right) \ket{i_1,\ldots,i_{N}}.
\end{equation}
For normal tensors, in fact, this is the only way how two different sets of tensors can generate the same states as stated by the fundamental theorem which was proven in Ref. \cite{Molnar2018}:
\begin{theorem}[\cite{Molnar2018}]
\label{thm:fund}
 The tensors $A_1,\ldots,A_N$ and $B_1,\ldots,B_N$ generate the same normal MPS $\Psi$ iff there exist regular matrices $x_1,\ldots,x_N$ such that $A_k^j = x_k^{-1}B_k^jx_{k+1}$ for all $k$ and $j$, with $x_{N+1} \equiv x_1$; that is, iff
\begin{align}
\ket{A_k} = \one \otimes x_k^{-1} \otimes x_{k+1}^T\ket{B_k} \ \forall \ k. \label{eq:fundtheom}
\end{align}
The matrices $x_1,\ldots,x_N$ are unique up to a multiplicative constant.
\end{theorem}
Theorem \ref{thm:fund} is the basis of the proofs of Theorem 1 and Theorem 2 of the main text.

\section{Proof of Theorem 1 and Theorem 2}
\label{sec:proofs}
In this section we provide the proof of Theorem 1 and Theorem 2 of the main text. To this end, the following Lemma will be important.

\begin{lemma}
\label{lem:stayinj}
	Suppose $\Psi\in \mathcal{N}_{N,D}$ defined by a tensor $A$ with injectivity length $L$ can also be written as a MPS with non-TI tensors $B_1,\dots B_N$, all with bond dimension $D$. Then this description is also normal with injectivity length $L$.
\end{lemma}
\begin{proof}
  The two different ways to express the MPS are the following:
  \begin{equation}
    \ket{\Psi} = \sum_i \tr(A^{i_1} \dots A^{i_{N}} ) \ket{i_1 \dots i_N} = \sum_i \tr(B_1^{i_1} \dots B_N^{i_{N}} ) \ket{i_1 \dots i_N} .
  \end{equation}
  Let us apply any linear functional to the last $N-k$ subsystems, where $k$ satisfies $L\leq k \leq N-L$. That is, we consider the vector space
  \begin{equation}
    V = \left\{\sum_i f_{i_{k+1} \dots i_{N}} \cdot \tr(A^{i_1} \dots A^{i_{k}} A^{i_{k+1}} \dots A^{i_{N}}) \ket{i_1 \dots i_k} \ \middle | \  f\in \mathbb{C}^{d(N-k)} \right\} .
  \end{equation}
  Due to the normality of the $A$ tensor, the products of the last $N-k\geq L$ matrices describing the MPS span the whole space of $D\times D$ matrices, that is,
  \begin{equation}
    \left\{\sum_i f_{i_{k+1} \dots i_{N}} \cdot A^{i_{k+1}} \dots A^{i_{N}} \ \middle | \  f\in \mathbb{C}^{d(N-k)} \right\}  = \mathbb{C}^{D\times D}.
  \end{equation}
  Therefore, the vector space $V$ can also be written as
  \begin{equation}
    V = \left\{\sum_i  \tr(A^{i_1} \dots A^{i_{k}} \cdot X ) \ket{i_1 \dots i_N} \ \middle | \  X\in \mathbb{C}^{D \times D} \right\} .
  \end{equation}
  Due to the normality of the tensor $A$, the map
  \begin{equation}
    X \mapsto \sum_i  \tr(A^{i_1} \dots A^{i_{k}} \cdot X ) \ket{i_1 \dots i_N}
  \end{equation}
  is injective. As it is also linear, the vector space  $V$ is $D^{2}$ dimensional. $V$ can also be expressed with the help of the $B$ tensors. Similarly to the derivation above, we find
  \begin{equation}
    V = \left\{\sum_i  \tr(B_1^{i_1} \dots B_k^{i_{k}} \cdot X ) \ket{i_1 \dots i_N} \ \middle | \  X\in W \leq \mathbb{C}^{D \times D} \right\} ,
  \end{equation}
  where $W$ is a subspace of the space of all $D$-by-$D$ matrices that is spanned by the products of the last $N-k$ matrices describing the MPS. As $V$ is $D^2$ dimensional, it immediately follows that $W=\mathbb{C}^{D \times D}$ and that the map
  \begin{equation}
    X \mapsto \sum_i  \tr(B_1^{i_1} \dots B_k^{i_{k}} \cdot X ) \ket{i_1 \dots i_N}
  \end{equation}
  is injective. This argument can be repeated to any $L$ consecutive subsystems, thus the tensors
  $B_1, \ldots B_N$ form a normal description of the MPS $\Psi$.
\end{proof}

Using Lemma \ref{lem:stayinj} we can proof Theorem 1 of the main text, which provides a characterization of the local symmetries of a normal MPS $\Psi(A)$, i.e., of all $S = s_1 \otimes \ldots \otimes s_N$ such that
\begin{align}
 S\ket{\Psi(A)} = \ket{\Psi(A)}. \label{eq:sym2}
\end{align}
We restate the theorem here for the sake of readability.\\

\noindent {\bf Theorem 1.} \hspace{-0.35cm}
\emph{
The local (global) symmetries of $\Psi(A)\in {\cal N}_{N,D}$ are in one-to-one correspondence with the $N$-cycles (1-cyles) in $G_A$.
}
\proof{
We first show that the $S = s_1 \otimes \ldots \otimes s_N$ that solve Eq. (\ref{eq:sym2}) correspond to $N$-cycles in $G_A$. Note that the state $S\ket{\Psi(A)}$ is an MPS with bond dimension $D$, generated by the fiducial states $\ket{A_k} = s_k \otimes \one \otimes \one\ket{A}$ for $k = 1,\ldots,N$. Lemma \ref{lem:stayinj} implies that the representation $S\ket{\Psi(A)}$ of the normal MPS $\ket{\Psi(A)}$ is normal too and thus Theorem \ref{thm:fund} can be used to find all $S$ that satisfy Eq. (\ref{eq:sym2}). Because of Theorem \ref{thm:fund}, Eq. (\ref{eq:sym2}) is fulfilled iff there are (up to a multiplicative factor) unique regular matrices $x_1,\ldots,x_N$ such that
\begin{align}
 (s_k \otimes x_k^{-1} \otimes x_{k+1}^T)\ket{A} = \ket{A} \ \forall k, \label{eq:sym3}
\end{align}
where $x_{N+1} = x_1$. That is, $S = s_1 \otimes \ldots \otimes s_N$ is a symmetry of $\Psi(A)$ iff there are operators $h_1,\ldots,h_N \in G_A$, with $h_k = s_k \otimes x_k \otimes y_k^T$ \footnote{Note that the local operators defining $h_k$ are, of course, only unique up to a multiplicative factor, i.e., $h_k = s_k \otimes x_k \otimes y_k^T = (\frac{1}{\lambda_1 \cdot \lambda_2} s_k) \otimes (\lambda_1 \cdot x_k) \otimes (\lambda_2 \cdot y_k)^T$ for any $\lambda_1, \lambda_2 \neq 0$. However, in the following we always fix one particular choice of local operators.}, that can be connected to an $N$-cycle, i.e., for which $y_kx_{k+1} \propto \one$ holds. This shows that the local symmetry group of $\Psi(A)$ is in one-to-one correspondence with the $N$-cycles in $G_A$.

If $S = s^{\otimes N}$ is a global symmetry, $s_k \propto s$ holds and thus the uniqueness (up to a multiplicative factor) of the $x_k$ matrices in Eq. (\ref{eq:sym3}) implies that they are all proportional to each other. Hence, a symmetry is global iff it originates from a 1-cycle. \qed \\
}

Theorem 2 provides a criterion for when the transformation $A \Nto B$ is possible, i.e., when there is a $g = g_1 \otimes \ldots \otimes g_n$ such that
\begin{align}
 g\ket{\Psi(A)} = \ket{\Psi(B)}. \label{eq:trans1}
\end{align}
We again restate the theorem before we prove it.\\

\noindent {\bf Theorem 2.} \hspace{-0.35cm}
\emph{
$A \Nto B$ with local (global) transformations iff there exists an $N$-cycle (1-cycle) in $G_{A,B}$.
}
\proof{
The ``if''-part is trivial. To prove the ``only if''-part suppose that Eq. (\ref{eq:trans1}) holds. Then $g\ket{\Psi(A)}$ is an MPS representation of $\Psi(B)$ with the same bond dimensions. Lemma \ref{lem:stayinj} then implies that $g\ket{\Psi(A)}$ is normal too; even if $g$ is singular. Hence, $g\ket{\Psi(A)}$ and $\Psi(B)$ have to be related as stated by the fundamental theorem, Theorem \ref{thm:fund}. Analogously to the proof of Theorem 1 one can use this to show that $g$ has to correspond to an $N$-cycle in $G_{A,B}$.
\qed
}

\section{Symmetries of examples in the main text}
\label{sec:syms}

In this section we derive the symmetries of the states presented in the main text. We denote the local symmetry group of $\Psi(A)$ by
\begin{align}
 S_{\Psi(A)} \equiv \{S = s_1 \otimes \ldots \otimes s_N \ | \ S\ket{\Psi(A)} = \ket{\Psi(A)}\}. \label{eq:symgroup}
\end{align}

\subsection{Symmetries of injective MPS}

For injective MPS we use decomposition (\ref{eq:dec1}) for the fiducial state and the fact that $\mathcal{A}^{-1}$ exists if the MPS is injective. Moreover, we use that the maximally entangled state defined after Eq. (\ref{eq:dec1}) satisfies the following equation for any $x$,
\begin{align}
 (\one \otimes x)\ket{\Phi^+} = (x^T \otimes \one)\ket{\Phi^+}.
\end{align}
Using these properties it is straightforward to verify that
\begin{align}
 G_A = \{s_{x,y} \otimes x \otimes y^T | x,y \in GL(D, \mathbb{C})\}, \label{eq:syminj1}
\end{align}
where $s_{x,y} = \mathcal{A}({x^T}^{-1} \otimes y^{-1})\mathcal{A}^{-1}$. Clearly, the symmetry $s_{x,y} \otimes x \otimes y^T$ can be connected to any symmetry $s_{y^{-1},z} \otimes y^{-1} \otimes z^T$, where $z \in GL(D,\C)$ is arbitrary. Using this in combination with Theorem 1 yields
\begin{align}
 S_{\Psi(A)} = \left\{s_{x_N^{-1},x_1} \otimes \ldots \otimes s_{x_{N-1}^{-1},x_N}\right\}_{x_1,\ldots,x_N \in GL(D,\C)}.
\end{align}

\subsection{Symmetries of the AKLT state}
The AKLT state is generated by the matrices \cite{Affleck1987}
\begin{align}
A_A^0 = \sqrt{2} \begin{pmatrix} 0 & 0 \\ 1 & 0 \end{pmatrix}, \ A_A^1 = \begin{pmatrix} -1 & 0 \\ 0 & 1 \end{pmatrix}, \  A_A^2 = \sqrt{2} \begin{pmatrix} 0 & -1 \\ 0 & 0 \end{pmatrix}.
\end{align}
We use that
\begin{align}
s \otimes x \otimes y^T \in G_{A_A} \ \Leftrightarrow \ x A_A^i y = \sum_{j = 0}^2(s^{-1})_{ij} A_A^j \ \forall i, \label{eq:AKLT1}
\end{align}
where $(s^{-1})_{ij}$ denotes the entries of $s^{-1}$. We can then take the trace on the right-hand side of Eq. (\ref{eq:AKLT1}) and use that the matrices $A_A^j$ are traceless to obtain the following equation,
\begin{align}
 \tr(A_A^i yx) = \sum_{j}(s^{-1})_{ij} \tr(A_A^j) = 0 \ \forall i. \label{eq:AKLT2}
\end{align}
Note further that $(\one, A_A^0, A_A^1, A_A^2)$ forms an orthogonal basis of all 2-by-2 matrices. Thus, Eq. (\ref{eq:AKLT2}) implies that $y = \frac{1}{\lambda} x^{-1}$ for some $\lambda \neq 0$. Inserting this into the right-hand side of (\ref{eq:AKLT1}) yields
\begin{align}
 x A_A^i x^{-1} = \lambda \sum_{j}(s^{-1})_{ij} A_A^j \label{eq:AKLT3}
\end{align}
We can absorb $\lambda$ in the definition of $s$ and thus set $\lambda = 1$, without loss of generality. Since $(A_A^0, A_A^1, A_A^2)$ is a basis of all traceless 2-by-2 matrices one can then find, for any regular $x$, a regular $s = s_x$ such that Eq. (\ref{eq:AKLT3}) holds.
Summarizing, this shows that $G_A = \{s_x \otimes x^{-1} \otimes x^T\}_{x \in GL(2,\C)}$. Using Theorem 1 then further shows that
\begin{align}
 S_{\Psi(A_A)} = \left\{s_x^{\otimes N} \right\}_{x \in GL(2,\C)}. \label{eq:AKLT4}
\end{align}

Note that the following observation holds.

\begin{observation}
 The symmetry group $S_{\Psi(A_A)}$ is isomporphic to the projective linear group $PGL(2,\C)$.
\end{observation}
\proof{
We have to show that the following is satisfied for any regular $x,y$
\begin{align}
 s_x = s_y \Leftrightarrow x \propto y.
\end{align}
From Eq. (\ref{eq:AKLT3}) (recall that we have set, w.l.o.g., $\lambda = 1$)  it is easy to see that $x \propto y$ implies $s_x = s_y$. To show that also the reverse holds suppose that $s_x = s_y$ holds for some regular $x,y$. Then Eq. (\ref{eq:AKLT3}) (again with $\lambda = 1$) implies that $x A_A^i x^{-1} = y A_A^i y^{-1}$ for all $i$, which is equivalent to 
\begin{align}
 y^{-1}x A_A^i = A_A^iy^{-1}x \ \forall i.
\end{align}
This shows that $y^{-1}x$ commutes with all $A_A^i$. Since $(\one, A_A^0, A_A^1, A_A^2)$ forms a basis of all 2-by-2 matrices this shows that $y^{-1}x \propto \one$. \qed
}

\subsection{Symmetries of AKLT-type states}
\label{sec:symAKLTtype}
The AKLT-type states are generated by the fiducial state $\ket{A_{A,g}} = (\one \otimes g \otimes \one) \ket{A_A}$, where $g \in GL(2,\C)$ is such that the resulting state is normal. As noted in the main text we have
\begin{align}
 G_{A_{A,g}} = (\one \otimes g \otimes \one) \cdot G_A \cdot (\one \otimes g^{-1} \otimes \one) = \{h_x \equiv s_x \otimes gx^{-1}g^{-1} \otimes x^T\}_{x \in GL(2,\C)}. \label{eq:defAKLT1}
\end{align}
Two operators $h_x,h_y \in G_{A_{A,g}}$ can be concatenated iff $y \propto g^{-1}xg$. Hence, the operators $h_{x_1},\ldots,h_{x_N}$ form an $N$-cycle iff
\begin{align}
x_{k+1} \propto g^{-1}x_k g \ \forall k
\end{align}
where $x_{N+1} \equiv x_1$. This is fulfilled for an $x \equiv x_1$ iff $x = g^{-N}xg^N$. Using Theorem 1 this yields
\begin{align}
S_{\Psi(A_{A,g})} \equiv \{s_x \otimes s_{g^{-1}xg} \otimes \ldots \otimes s_{g^{-(N-1)}xg^{N-1}} \ | \ x \in GL(2,\C), x = g^{-N}xg^N\}.
\end{align}

\subsection{Cluster state and W-generated states}
We refer the reader to Section \ref{sec:SLOCC}, where we characterize the SLOCC classes and the local symmetries of all normal MPS with $d = D = 2$.

\section{Transformations of examples in the main text}

In this section we derive the transformations of the states presented in the main text.

\subsection{From Injective MPS to other MPS}
We again use decomposition (\ref{eq:dec1}) and the fact that $\mathcal{A}^{-1}$ exists for injective MPS. For an injective MPS $\Psi(A)  \in \mathcal{N}_{N,D}$ and an arbitrary $\Psi(B) \in \mathcal{N}_{N,D}$ it is then  straightforward to see that $(\mathcal{B}\mathcal{A}^{-1} \otimes \one \otimes \one) \in G_{A,B}$ forms a $1$-cycle and thus the transformation $A \Nto B$ can be achieved via a global operation as
\begin{align}
 (\mathcal{B}\mathcal{A}^{-1})^{\otimes N} |\Psi(A)\rangle = |\Psi(B)\rangle.
\end{align}
Combining Theorem 1 with the fact that $G_{A,B} = (\mathcal{B}\mathcal{A}^{-1} \otimes \one \otimes \one) \cdot G_A$, where $G_A$ is given in Eq. (\ref{eq:syminj1}), it is easy to see that there are also infinitely many non-TI  operations that achieve the transformation $A \Nto B$.

\subsection{From the AKLT state to the cluster state}
\label{sec:trafos}
Let us now determine when the AKLT state can be transformed into the cluster state.
The cluster state is generated by the fiducial state $\ket{A_{Cl}} = (\one \otimes H \otimes \one) \ket{GHZ}$, where $H = \sum_{i,j = 0}^1 (-1)^{ij} \ket{i}\bra{j}$ and $\ket{GHZ} = \frac{1}{\sqrt{2}}(\ket{000} + \ket{111})$ is the three-qubit GHZ state. Note that we can write
\begin{align}
 G_{A_A,A_{Cl}} = (\one \otimes H \otimes \one) \cdot G_{A_A,GHZ}. \label{eq:AKLTcluster1}
\end{align}
Let us first determine $G_{A_A,GHZ}$. To this end, note that $G_{A_A,GHZ} \subset \C^{2 \times 3} \otimes GL(2,\C) \otimes GL(2,\C)$, where we have used that $A_A$ and the GHZ state are both tripartite entangled and, therefore, the operators on the bond dimensions have to be invertible. Note further that any $x \in \C^{2 \times 3}$ can be expressed as $x = z M$, where $z \in GL(2,\C)$ and $M$ is a 2-by-3 matrix in reduced row-echelon form \cite{Horn2013}, i.e., is an element of the set
\begin{align}
 E_{2,3} \equiv \left\{M_1(\alpha,\beta) \equiv \begin{pmatrix} 1 & 0 & \alpha\\ 0 & 1 & \beta \end{pmatrix}\right\}_{\alpha,\beta \in \C} \ \cup \ \left\{M_2 \equiv \begin{pmatrix} 1 & 0 & 0 \\ 0 & 0 & 1 \end{pmatrix}, \ M_3 \equiv \begin{pmatrix} 0 & 1 & 0 \\ 0 & 0 & 1 \end{pmatrix} \right\}.
\end{align}
Hence, we can write $h \in G_{A_A,GHZ}$ as $h = w(M \otimes \one \otimes \one)$, where $w \in GL(2,\C)^{\otimes 3}$ and $M \in E_{2,3}$. A necessary condition for $h \in G_{A_A,GHZ}$ obviously is that $h\ket{A_A}$ is a state in the SLOCC class of the GHZ state.  Recall that a general three-qubit state $\ket{\psi} = \ket{0}\ket{\phi_0} + \ket{1}\ket{\phi_1}$, with $\ket{\phi_i} \in \C^2 \otimes \C^2$, is an element of the GHZ class iff its three-tangle is non-vanishing \cite{Coffman2000}, i.e., iff
\begin{align}
 \tau_3(\psi) = \left \lvert \text{det}\begin{pmatrix} \bra{\phi_0^*}\sigma_2^{\otimes 2} \ket{\phi_0} & \bra{\phi_0^*}\sigma_2^{\otimes 2} \ket{\phi_1}\\
 \bra{\phi_1^*}\sigma_2^{\otimes 2} \ket{\phi_0} & \bra{\phi_1^*}\sigma_2^{\otimes 2} \ket{\phi_1} \end{pmatrix}\right \rvert \ \neq 0.
\end{align}
Here, $\ket{\phi^*}$ denotes the complex conjugate of the state $\ket{\phi}$ in the computational basis. Moreover, for any $t \in GL(2,\C)^{\otimes 3}$ we have that $\tau(t\ket{\psi}) \neq 0$ iff $\tau(\psi) \neq 0$.

Hence, $h = w(M \otimes \one \otimes \one) \in G_{A_A,GHZ}$ has to fulfill
\begin{align}
 \tau(M \otimes \one \otimes \one \ket{A_A}) \neq 0. \label{eq:tangle1}
\end{align}
Inequality (\ref{eq:tangle1}) is satisfied iff $M = M_1(\alpha,\beta)$ with $\alpha \neq -\frac{\beta^2}{2}$ or $M = M_2$. In particular, it is not fulfilled for $M = M_3$, such that we no longer have to consider this case. For matrices $M$ that fulfill inequality (\ref{eq:tangle1}) it is straightforward to find a $w \in GL(2,\C)^{\otimes 3}$ such that  $h = w(M \otimes \one \otimes \one) \in G_{A_A,GHZ}$. In this way, we arrive at the following operators of $G_{A_A,GHZ}$,
\begin{align}
&h_1(\alpha,\beta) = a(\alpha,\beta)M_1(\alpha,\beta) \otimes b(\alpha,\beta) \otimes c(\alpha,\beta), \ \text{for} \ \alpha \neq -\frac{\beta^2}{2},\\
&h_2 = \frac{1}{\sqrt{2}} M_2 \otimes \sigma_1 \otimes \sigma_3,
\end{align}
where
\begin{align*}
&a(\alpha,\beta) \equiv
 \begin{pmatrix}
 1 & \beta-\sqrt{2 \alpha + \beta^2} \\
 1 & \beta +\sqrt{2 \alpha +\beta^2} \\
\end{pmatrix},\\
&b(\alpha,\beta) \equiv
\begin{pmatrix}
 -\frac{1}{2\sqrt{2}} \frac{1}{2\alpha + \beta^2} & \frac{1}{4} \frac{-\beta - \sqrt{2\alpha + \beta^2}}{2\alpha + \beta^2} \\
 1 & \frac{1}{\sqrt{2}} (b - \sqrt{2\alpha + \beta^2}) \\
\end{pmatrix},\\
&c(\alpha,\beta) \equiv
\begin{pmatrix}
 \frac{1}{\sqrt{2}} (-b - \sqrt{2\alpha + \beta^2}) & 1 \\
 \frac{1}{4} \frac{\beta - \sqrt{2\alpha + \beta^2}}{2\alpha + \beta^2} & -\frac{1}{2\sqrt{2}} \frac{1}{2\alpha + \beta^2} \\
\end{pmatrix}.
\end{align*}
We obtain the whole set $G_{A_A,GHZ}$ by multiplying these operators from the left with the local symmetry group of the GHZ state, which reads \cite{Verstraete2002b}
\begin{align}
 G_{GHZ} = \left\{s_{GHZ}^{(i,x,y)} = \sigma_1^i P_{\frac{1}{xy}} \otimes \sigma_1^{i} P_{x} \otimes \sigma_1^iP_{y} \right\}_{(i,x,y) \in \{0,1\} \times {\C^{\times}}^2} \label{eq:GHZsymm},
\end{align}

with $P_z = \text{diag}(z,1/z)$. Combining this with Eq. (\ref{eq:AKLTcluster1}) we obtain,
\begin{align}
 G_{A_A,A_{Cl}} = G_{A_A, Cl}^{(1)} \cup  G_{A_A, Cl}^{(2)},
\end{align}
with
\begin{align}
 G_{A_A, Cl}^{(1)} = (\one \otimes H \otimes \one)\cdot G_{GHZ} \cdot \left\{h_1(\alpha,\beta) \ | \ \alpha \neq -\frac{\beta^2}{2}\right\}, \ G_{A_A,Cl}^{(2)} = (\one \otimes H \otimes \one) \cdot G_{GHZ} \cdot h_2.
\end{align}
Due to Theorem 2 it now only remains to determine the cycles that can be obtained by concatenating elements of $G_{A_A,A_{Cl}}$. It is straightforward to see that elements of $G_{A_A, A_{Cl}}^{(k)}$, for $k \in \{1,2\}$, cannot be concatenated with each other. However, an element of $G_{A_A, A_{Cl}}^{(1)}$ can be concatenated with an element of $G_{A_A, A_{Cl}}^{(2)}$ to form a 2-cycle. The only way to obtain an $N$-cycle is therefore to alternatingly concatenate elements from $G_{A_A, A_{Cl}}^{(1)}$ and $G_{A_A, A_{Cl}}^{(2)}$; which is possible iff $N$ is even. This proves that $A_A \Nto A_{Cl}$ iff $N$ is even.

Let us note that the method presented here can also be used to determine \emph{all} MPS with $d = D = 2$ to which the AKLT state can be transformed.

\subsection{From the AKLT state to AKLT-type states}
Let us determine when the transformation $A_A \Nto A_{A,g}$ from the AKLT state to an AKLT-type state is possible. Note first that
\begin{align}
 G_{A_A,A_{A,g}} = (\one \otimes g \otimes \one) \cdot G_{A_A} = \{h_x = s_x \otimes gx^{-1} \otimes x^T | x \in GL(2,\C\},
\end{align}
where $s_x$ was defined in Section \ref{sec:symAKLTtype}. The operators $h_{x_1},\ldots,h_{x_N}$ form an $N$-cycle iff
\begin{align}
x_{k+1} \propto x_kg \ \forall k,
\end{align}
where $x_{N+1} \equiv x_1$. This is fulfilled for any $x_1$ iff $g^N \propto \one$. Using Theorem 1 we see that the following holds.
\begin{align}
 A_{A_A} \Nto A_{A,g} \ \Leftrightarrow \ g^N \propto \one
\end{align}

\section{Symmetries and SLOCC classification for MPS with $d = D = 2$}
\label{sec:SLOCC}
It is straightforward to show that MPS generated by (bi-)separable three-qubit states are product states (i.e. they have bond dimension $D = 1$). Hence, we only have to consider MPS generated by genuinely tripartite entangled three-qubit states, which are either an element of the GHZ class or the W class \cite{Dur2000}. As explained in the main text, it is sufficient to determine when normal MPS generated by fiducial states of the form
\begin{align}
&\ket{GHZ_b} = \one \otimes b \otimes \one\ket{GHZ}, \ \text{i.e., with matrices} \ A_{GHZ,b}^0 = b\ket{0}\bra{0}, \ A_{GHZ,b}^1 = b\ket{1}\bra{1}, \ \text{or} \label{eq:simplefid1}\\
&\ket{W_b} = \one \otimes b \otimes \one\ket{W}, \ \text{i.e., with matrices} \ A_{W,b}^0 = b(\ket{0}\bra{1} + \ket{1}\bra{0}), \ A_{W,b}^1 = b\ket{0}\bra{0} \label{eq:simplefid2},
\end{align}
 are related via transformations that are not global. The whole classification is obtained by adding the states that are related to those states via global operations.

In order to characterize the local symmetry group of all normal MPS (see Eq. (\ref{eq:symgroup})) we can use the following property. For $A \Nsim B$ there exists, by definition, an invertible local operator $g$ such that $\ket{\Psi(B)} = g\ket{\Psi(A)}$ and it is straightforward to see that
\begin{align}
S_{\Psi(B)} = g S_{\Psi(A)} g^{-1}. \label{eq:conjsym}
\end{align}
Hence, it is sufficient to find the symmetries of one representative of an SLOCC class, $\Psi(A)$. Concretely, this means that it is sufficient to characterize the symmetries of MPS of the form (\ref{eq:simplefid1} - \ref{eq:simplefid2}).

In order to find the symmetries and SLOCC classes of these MPS we proceed in three steps:
\begin{enumerate}
 \item Determine for which $b$ the state $\Psi(X_b)$ is normal.
 \item Characterize the symmetries of the normal MPS using Theorem 1.
 \item Characterize the SLOCC classes of the states $\Psi(X_b)$ using Corollary 1 of the main text. To simplify this procedure, we can use that MPS with different injectivity lengths cannot be SLOCC equivalent (this follows from Lemma \ref{lem:stayinj}). Moreover, MPS whose symmetry groups are not conjugate to each other, i.e., do not fulfill Eq. (\ref{eq:conjsym}) for some $g$, can also never be SLOCC equivalent.
\end{enumerate}

\begin{table}[h!]
 \begin{tabular}{c|c|c|c|c}
 type &  \# symmetries & inj. length & \# SLOCC classes & $A \Nsim B$\\
 \hline
  GHZ &$2$ &  $2$ & $\infty$ (generic set) & $GHZ_b \Nsim GHZ_c \ \Leftrightarrow \chi(b) = \begin{cases} \chi(c) \ \text{or} \ \chi(c)^{-1}, \ \text{$N$ even}\\
\chi(c), \ \text{$N$ odd}.\end{cases}$\\ \cline{2-5}
  & $2^N$  & $2$ & $1$ (cluster set) & always \\ \cline{2-5}
  & $1$ &$3$ & $1$ (symmetryless set) & always\\ \cline{2-5}
  & $\infty$ & not normal & $1$ (GHZ$_N$ class) & always\\
 \hline
 W & $1$ for odd $N$& $2$ & $1$ & always\\
  & $\infty$ for even $N$ & & &
 \end{tabular}
 \caption{The SLOCC classification of MPS with $d = D = 2$. First, according to the SLOCC class of the generating three-qubit state. For GHZ-generated states one can further coarse grain the classes according to their local symmetries into different sets. We also provide the minimal number of qubits that have to be blocked to make the normal states injective. The only multipartite entangled non-normal states are all SLOCC equivalent to the non-normal state generated by the three-qubit GHZ state, i.e., they are elements of the GHZ$_N$ class. We state how many different SLOCC classes there are within one set and depict when two MPS within this set are SLOCC equivalent. The function $\chi$ is defined in Eq. (\ref{eq:chi}) (see also main text). Note that the class with two local symmetries is of full measure in the set of all MPS with $d = D = 2$.}
  \label{tab:2}
\end{table}

The resulting symmetry characterization and SLOCC classification is concisely summarized in Table \ref{tab:2}, which is an extended version of Table I in the main text. Let us note here that $\Psi(GHZ_{\one}) = \ket{GHZ_N} \equiv \frac{1}{\sqrt{2}} (\ket{0}^{\otimes N} + \ket{1}^{\otimes N})$ is the $N$-qubit GHZ state. This state is not normal and thus the methods of the main text do not directly apply to it. However, the symmetries of $GHZ_N$ are known \cite{Verstraete2002b}. Moreover, we show below that all non-normal multipartite entangled MPS are SLOCC equivalent to $GHZ_N$. Although the SLOCC and symmetry classification of general non-normal MPS is not within the scope of the main text, we could thus determine it for the special case of $d = D = 2$. Combined with the results on normal MPS we therefore obtain here a characterization of the symmetries and SLOCC classes of all multipartite entangled MPS with $d = D = 2$.\\

In the following we provide a detailed presentation and derivation of these results. We first consider the GHZ- (Section \ref{sec:GHZgenerated}) and then the W-generated states (Section \ref{sec:Wgenerated}).

\subsection{GHZ-generated MPS}
\label{sec:GHZgenerated}
\subsubsection{Characterization of the normal MPS}
\label{sec:normalGHZ}
We first determine when $\Psi(GHZ_b)$ is normal, where $b = (b_{ij}) \in GL(2,\C)$. That is, we have to check for which $b \in GL(2,\C)$ there is an $L$ such that
\begin{align}
 \underset{i_1,\dots,i_L}{\mathrm{span}}\left\{A_{GHZ,b}^{i_1}\cdot \ldots \cdot A_{GHZ,b}^{i_L}\right\} = \C^{2 \times 2}.
\end{align}
Here, we determine the minimal $L$ with this property, i.e., the injectivity length of $\Psi(W_b)$. It is straightforward to see that we have to distinguish four different cases:
\begin{itemize}
 \item[(i)] $b_{ij} \neq 0$ for all $i,j$: $L = 2$ and thus the MPS is normal for $N \geq 5$. Note that the states related to MPS of this case via (trivial) global operations are generated by fiducial states of the form $g \otimes x^{-1}b \otimes x^T\ket{GHZ}$ (as shown in the main text),  where $g,x$ are arbitrary regular matrices. Since $b$ is a generic regular matrix (for this case) these fiducial states comprise a generic set of three-qubit states. Hence, the MPS corresponding (up to global operations) to this case are generated by a full measure set of three-qubit states and are thus of full measure in the set of all MPS with $d = D = 2$.
 \item[(ii)] exactly one entry of $b$ is zero:
	    \begin{itemize}
	      \item[(iia)]$b_{kk} = 0$ for exactly one $k \in \{0,1\}$: $L = 3$ and thus the MPS is normal for $N \geq 7$,
	      \item[(iib)]$b_{01} = 0$ or $b_{10} = 0$: The MPS is not normal for any $N$ and SLOCC equivalent to $\ket{GHZ_N}$.
	    \end{itemize}
 \item[(iii)] exactly two entries of $b$ are zero: The MPS is either SLOCC equivalent to $\ket{GHZ_N}$ or vanishes and is therefore not normal.
\end{itemize}

In particular, this shows that normal GHZ-generated MPS have an injectivity length of at most 3 (in case (iia)) and generically (i.e., in case (iiia)) of 2. This is considerably below the best known upper bound (to the knowledge of the authors) of $L \leq 2D^2(6+\log_2(D))$ for the injectivity length of a normal MPS with physical dimension $d$ and bond dimension $D$ \cite{Michalek2018}. For $D = 2$ this bound states $L \leq 56$.

\subsubsection{Characterization of the local symmetries}
\label{sec:stabGHZ}
In the following we determine the local symmetries of the normal GHZ-generated MPS determined before (i.e., of states belonging to the cases (i) and (iia) in the last section). Note that the symmetries of the three-qubit GHZ state are given in Eq. (\ref{eq:GHZsymm}). The stabilizer of the GHZ-type state $\ket{GHZ_b} = \one \otimes b \otimes \one \ket{GHZ}$ hence reads
\begin{align*}
G_{GHZ_b} = \left\{s^{(k,x,y)} =  (\one \otimes b \otimes \one)s_{GHZ}^{(k,x,y)}(\one \otimes b^{-1} \otimes \one)\right\}_{(k,x,y) \in \{0,1\} \times \C^\times}.
\end{align*}
Two elements  $s^{(k,v,w)}, s^{(l,x,y)} \in G_{GHZ_b}$ can be concatenated iff
\begin{align}
P_{w}\sigma_1^k b\sigma_1^{l} P_{x}b^{-1} = r \one, \label{eq:dir3qbGHZ}
\end{align}
for some $r \neq 0$, where $P_z = \text{diag}(z,1/z)$.
This condition is extremely restrictive and it is easy to find the minimal cycles in $G_{GHZ_b}$ entailed by it. We can simply read off the resulting symmetries (as explained in the main text). This yields the following stabilizer for the cases (i) and (iia) found in Section \ref{sec:normalGHZ}.
\begin{itemize}
 \item[(i)] In solving Eq. (\ref{eq:dir3qbGHZ}) the function
 \begin{align}
 \chi(b) = \frac{b_{00} \cdot b_{11}}{b_{01} \cdot b_{10}}. \label{eq:chi}
 \end{align}
 plays a prominent role. More precisely, $\chi$ can be used to further distinguish the MPS in this case according the following subcases:
 \begin{itemize}
 \item[(ia)] $\chi(b) \neq -1,0$: Then Eq. (\ref{eq:dir3qbGHZ}) only has solutions if $k = l$ and they depend on $b$. For $k = l = 0$ we get $w = x = \pm 1$. For $k = l = 1$ we get $w^2 = \frac{b_{00}b_{01}}{b_{10}b_{11}}$ and $x =  \frac{b_{10}}{b_{01}} w$, $r = 1$. There is only one nontrivial cycle in $G_{GHZ_b}$, which has length 1. Hence, besides the trivial symmetry, the state $\Psi(GHZ_b)$ has one nontrivial symmetry and its stabilizer reads
 \begin{align}
  S_{\Psi(GHZ_b)} = \left\{\one^{\otimes N}, \left(\sigma_1P_{\frac{b_{11}}{b_{00}}}\right)^{\otimes N}\right\},
 \end{align}
 for $N \geq 5$.
 \item[(ib)] $\chi(b) = -1$: Equation (\ref{eq:dir3qbGHZ}) has the following solutions: $r = w = x = 1$ for $k = l = 0$; $r = i, w = i, x = \frac{b_{00}}{b_{01}}$ for $k = 0, l= 1$; $r = i, w = \frac{b_{00}}{b_{10}}, x = 1$ for $k = 1, l = 0$; $r = 1, w = \frac{i b_{00}}{b_{10}}, x = \frac{b_{10}}{b_{01}}$ for $k = l = 1$. Hence, there are many ways to connect elements in $G_{GHZ_b}$. They give rise to $2^N$ different $N$-cycles that each lead to a local symmetry of $\ket{\Psi(GHZ_b)}$. Note that the linear cluster state (with periodic boundary conditions) reads $\ket{Cluster} \equiv \ket{\Psi(GHZ_H)}$, where $H = \sum_{i,j = 0}^1 (-1)^{ij} \ket{i}\bra{j}$. For the cluster state we find that the local symmetries are exactly given by its stabilizer symmetries \cite{Hein2006}, i.e.,
 \begin{align}
  S_{Cluster} = S_{\Psi(GHZ_H)} = \left\{K_1^{i_1} \cdot \ldots \cdot K_N^{i_N}\right\}_{i_1,\ldots,i_N \in \{0,1\}}. \label{eq:symcluster}
 \end{align}
Here, $K_i= \sigma_3^{(i-1)} \sigma_1^{(i)} \sigma_3^{(i+1)}$ acts as $\sigma_1$ on qubit $i$ and as $\sigma_3$ on qubits $i-1, i$ (with periodic boundary conditions) and as the identity on all other qubits.
In fact, we see in Section \ref{sec:GHZSLOCC} below that all states with $\chi(b) = -1$ are SLOCC equivalent, such that we call this set of states the \emph{cluster set}. The symmetries of all states in this set can thus also be easily obtained from the symmetries (\ref{eq:symcluster}) of the cluster state via Eq. (\ref{eq:conjsym}).
 \end{itemize}
 \item[(iia)] These states fulfill $\chi(b) = 0$. There only exists a solution of Eq. (\ref{eq:dir3qbGHZ}) for $k = l = 0$ and $w,x = \pm 1$. This results in a trivial stabilizer, i.e.,
 \begin{align}
  S_{\Psi(GHZ_b)} = \{\one\},
 \end{align}
 for $N \geq 7$ (as the injectivity length of these states is $L = 3$).
\end{itemize}

\subsubsection{Characterization of the SLOCC classes}
\label{sec:GHZSLOCC}
From the results of the previous section we conclude that normal GHZ-generated states can be separated into three different sets according to their symmetries, where states from different sets are in different SLOCC classes:
\begin{enumerate}
\item $\chi(b) \neq -1,0$ (case (ia) of Section \ref{sec:stabGHZ}): These states have only 1 nontrivial symmetry, which is global. We call this set the generic set as it contains almost all MPS.
\item $\chi(b) = -1$ (case (ib) of Section \ref{sec:stabGHZ}): These states have $2^N$ symmetries. This set contains the cluster state and thus we refer to it as the cluster set.
\item $\chi(b) = 0$ (case (iia) of Section \ref{sec:stabGHZ}): These states have only the trivial symmetry and thus we refer to this set as the symmetryless set.
\end{enumerate}
In the following we determine the SLOCC classes within these sets.
Using the symmetries (\ref{eq:GHZsymm}) of the GHZ state and Corollary 1 of the main text this is straightforward and reveals the following SLOCC classification within the sets 1. to 3.:
\begin{enumerate}
 \item First, we determine when $\Psi(GHZ_b)$ and $\Psi(GHZ_c)$ (with $\chi(b),\chi(c) \not\in \{-1,0\}$) are related via a (trivial) global operation. This is the case iff $G_{GHZ}$ contains a $(b\rightarrow c)$-1-cycle. It is straightforward to show that this condition is satisfied iff $\chi(b) = \chi(c)$. Next, we have to determine the MPS that are related via $(b\rightarrow c)$-$N$-cycles with $N > 1$. To this end, we use the procedure explained in the paragraph after Corollary 1 in the main text.
 For two operators
 \begin{align}
  h_1 = g_1 \otimes x_1 \otimes y_1^T \equiv s_{GHZ}^{(k,v_1,v_2)}  \in G_{GHZ}, \label{eq:op1}\\
  h_2 = g_2 \otimes x_2 \otimes y_2^T \equiv s_{GHZ}^{(l,w_1,w_2)} \in G_{GHZ} \label{eq:op2}
 \end{align}
 we define the matrix,
 \begin{align}
  M \equiv y_{\alpha}^{-1} y_\beta \otimes (x_{\alpha} x_\beta^{-1})^T = (P_{v_2}\sigma_1^{k})^{-1} P_{w_2}\sigma_1^l \otimes \left [\sigma_1^k P_{v_1} (\sigma_1^l P_{w_1})^{-1}\right]^T.
 \end{align}
 As explained in the main text, $h_1,h_2$ form a $(b\rightarrow c)$-$2$-cycle iff there exists a $\lambda \neq 0$ such that
 \begin{align}
  M\vec{c} = \lambda \vec{c}. \label{eq:EV}
 \end{align}
 For any $c$ that solves Eq. (\ref{eq:EV}) we can find the corresponding $b$ as
 \begin{align}
  b = y_1cx_2 = P_{v_2} \sigma_1^k c \sigma_1^l P_{w_1},
 \end{align}
 as explained in the main text. In this way, we can show that $\Psi(GHZ_b), \Psi(GHZ_c)$ are related via a nontrivial $(b\rightarrow c)$-$2$-cycle iff $\chi(b) = \frac{1}{\chi(c)}$. Analogously, we can show that $\Psi(GHZ_b), \Psi(GHZ_c)$ are not related via a ($b \rightarrow c)$-$N$-cycle of any size if neither $\chi(b) = \chi(c)$ nor $\chi(b) = \frac{1}{\chi(c)}$ hold. Summarizing, we have just shown the following,
\begin{align}
GHZ_b \Nsim GHZ_c \ \Leftrightarrow \chi(b) = \begin{cases} \chi(c) \ \text{or} \ \chi(c)^{-1} \ \text{and $N$ even} \\
\chi(c) \ \text{and $N$ odd}.\end{cases} \label{eq:char1}
\end{align}
In particular, there are infinitely many, $N$-dependent SLOCC classes within this generic set.\\

Let us also briefly outline an alternative way to derive Eq. (\ref{eq:char1}). For all fixed pairs of matrices $b, c$ (with $\chi(b), \chi(c) \not\in \{-1,0\}$) one could explicitly determine all $h_1,h_2 \in G_{A} $ as in Eqs. (\ref{eq:op1} - \ref{eq:op2}) that are $(b \rightarrow c)$-connected. Note that
\begin{align}
 h_1 \xrightarrow{b \rightarrow c} h_2 \ \Leftrightarrow \ y_2bx_2 \propto c \ \Leftrightarrow \ P_{v_2}\sigma_1^{k} \cdot c \cdot \sigma_1^l P_{w_1} \propto b. \label{eq:connect1}
\end{align}
For fixed $b,c$ there are only very few or no $h_1, h_2$ that solve Eq. (\ref{eq:connect1}). For the $b,c$ for which there are elements of $G_A$ that can be $(b\rightarrow c)$-connected it is then straightforward to find all $(b\rightarrow c)$-$N$-cycles for $N = 1,2$ and show that there are no larger cycles. This then leads to Eq. (\ref{eq:char1}).

 \item All MPS in this symmetry class are related to each other via (trivial) 1-cycles and are thus SLOCC equivalent for any $N$.
 \item All MPS in this symmetry class are related to each other via (trivial) 1-cycles and are thus SLOCC equivalent for any $N$.
\end{enumerate}

\subsection{MPS generated by W-type states}
\label{sec:Wgenerated}

\subsubsection{Characterization of the normal MPS}
\label{sec:normalW}
Analogously to the GHZ case, we first have to determine when $\Psi(W_b)$ is normal. That is, we have to check for which $b \in GL(2,\C)$ there is an $L$ such that
\begin{align}
 \underset{i_1,\dots,i_L}{\mathrm{span}}\left\{A_{W,b}^{i_1}\cdot \ldots \cdot A_{W,b}^{i_L}\right\} = \C^{2 \times 2}.
\end{align}
Here, we determine the minimal $L$ with this property, i.e., the injectivity length of $\Psi(W_b)$. A straightforward calculation shows that the following cases have to be distinguished:
\begin{itemize}
 \item[(i)] $b_{ij} \neq 0$ for all $i,j$: $L = 2$ and thus $\Psi(W_b)$ is normal for $N \geq 5$.
 \item[(ii)] exactly one entry of $b$ is zero:
      \begin{itemize}
	  \item[(iia)] $b_{00} = 0$: $\Psi(W_b) \propto \ket{0}^{\otimes N}$ and, thus, these states are not normal.
	  \item[(iib)] else: $L= 2$ and thus $\Psi(W_b)$ is normal for $N \geq 5$.
      \end{itemize}
 \item[(iii)] exactly two entries of $b$ are zero:
	  \begin{itemize}
	  \item[(iiia)] $b_{01}, b_{10} = 0$: $L = 2$ and thus $\Psi(W_b)$ is normal for any $N \geq 5$.
	  \item[(iiib)] else: $\Psi(W_b)$ is a product state and therefore not normal.
	  \end{itemize}
\end{itemize}

\subsubsection{Characterization of the local symmetries}
The local symmetries of the W state are given by \cite{deVicente2013}
\begin{align}
 S_{W} = \left\{\frac{1}{x} \begin{pmatrix} x & -y-z\\ 0 & \frac{1}{x}\end{pmatrix} \otimes \begin{pmatrix} x & y\\ 0 & \frac{1}{x}\end{pmatrix} \otimes \begin{pmatrix} x & z\\ 0 & \frac{1}{x}\end{pmatrix} \right\}_{(x,y,z) \in {\C^\times}^3} \label{eq:Wsymm},
\end{align}
For the cases (i), (iib) and (iiia) of normal states determined in the last section, $G_{W_b}$ contains the trivial cycle (from $\one$ to $\one$) and a continuous set of nontrivial 2-cycles. Hence, $\Psi(W_b)$ has only the trivial symmetry if $N$ is odd and infinitely many symmetries if $N$ is even. Interestingly, these symmetries have the same form for all W-generated normal MPS, namely
      \begin{align}
       S_{\Psi(W_b)} = \begin{cases} \one  & \text{if $N$ is odd},\\ \left\{\left(z(x) \otimes z(\frac{1}{x})\right)^{\otimes \frac{N}{2}}\right\}_{x \in {\C^\times}}  & \text{if $N$ is even}, \end{cases}
      \end{align}
      where
      \begin{align}
       z(x) = \begin{pmatrix} x & (x-\frac{1}{x}) \frac{b_{01} + b_{10}}{b_{00}}\\ 0 & \frac{1}{x} \end{pmatrix}.
      \end{align}

\subsection{Characterization of the SLOCC classes}
The order of the symmetries of normal W-generated states cannot be used to distinguish SLOCC classes. The reason for this is that all such states are SLOCC equivalent. To see this, we consider the normal MPS $\Psi(W) = \Psi(W_{\one})$ and an arbitrary normal MPS $\Psi(W_c)$. Then $G_A$ contains a $(\one \rightarrow c)$-1-cycle for any such choice of $c$. Hence, $\Psi(W) \Nsim \Psi(W_c)$ holds for any $N$ and $\Psi(W)$ can be transformed to any other normal W-generated MPS via a (trival) global transformation. Consequently, all normal W-generated MPS are in the same SLOCC class.

\bibliography{BibLetter}